\documentclass[runningheads]{llncs}

\usepackage{amsmath,amsfonts}
\usepackage{algorithm}
\usepackage{array}
\usepackage[caption=false,font=normalsize,labelfont=sf,textfont=sf]{subfig}
\usepackage{textcomp}
\usepackage{stfloats}
\usepackage{url}
\usepackage{verbatim}
\usepackage{graphicx}
\usepackage{cite}
\usepackage[utf8]{inputenc}
\usepackage[english]{babel}
\usepackage[T1]{fontenc}
\usepackage{cryptocode}
\usepackage{hyperref}
\hypersetup{
    colorlinks,%
    citecolor=black,%
    filecolor=black,%
    linkcolor=black,%
    urlcolor=black %
}
\usepackage{multirow}

\newcommand{\algsize}{\scriptsize}
\usepackage{amssymb}
\usepackage{graphicx}
\usepackage{latexsym}
\usepackage{threeparttable}
\usepackage{booktabs}
\usepackage{bbding}
\usepackage{epstopdf}
\usepackage{makecell}
\usepackage{colortbl}
\usepackage{rotating}
\usepackage{tabularx}
\usepackage{graphicx}
\usepackage{amsmath}
\usepackage{makeidx}  
\usepackage{color}    
\usepackage{epsfig}
\usepackage{pifont}
\usepackage{xfrac}
\usepackage{bbm}
\usepackage{multirow}
\usepackage{enumerate}
\usepackage{float}
\usepackage{algorithm}
\usepackage{algpseudocode}

\usepackage{comment}
\usepackage{multicol}
\usepackage{varwidth}
\usepackage{algpseudocode}
\usepackage{caption}
\usepackage{bbm}
\usepackage{epsfig}
\usepackage{xfrac}
\newcommand\Algphase[1]{%
\vspace*{-.7\baselineskip}\Statex\hspace*{\dimexpr-\algorithmicindent-2pt\relax}\rule{8.9cm}{0.4pt}%
\Statex\hspace*{-\algorithmicindent}\textbf{#1}%
\vspace*{-.7\baselineskip}\Statex\hspace*{\dimexpr-\algorithmicindent-2pt\relax}\rule{8.9cm}{0.4pt}%
}

\usepackage{stfloats}

\newcommand{\ceil}[1]{\ensuremath{\left\lceil #1 \right\rceil}}

\usepackage{fancybox}

\usepackage[draft,inline,nomargin]{fixme}
\definecolor{mygray}{gray}{.9}
\fxusetheme{color}

\begin{document}

\title{A Multi-Client Searchable Encryption Scheme for IoT Environment*}
%
%
\author{Nazatul H. Sultan \inst{1,2}\and Shabnam Kasra-Kermanshahi \inst{3} \and Yen Tran \inst{4}\and Shangqi Lai \inst{5}\and Vijay Varadharajan \inst{1}\and Surya Nepal \inst{2}\and Xun Yi \inst{3}}

%
\authorrunning{N. Sultan et al.}
%
\institute{University of Newcastle, Australia, \email{\{Nazatul.Sultan, Vijay.Varadharajan\}@newcastle.edu.au}
\and 
 CSIRO Data61, Australia, \email{Surya.Nepal@data61.csiro.au}
\and RMIT University, Australia,\email{\{shabnam.kasra.kermanshahi,xun.yi\}@rmit.edu.au} 
 \and
 UNSW Canberra, Australia, \email{ hongyen.tran@student.adfa.edu.au}
\and Monash University, Australia, \email{shangqi.lai@monash.edu}}


%
\maketitle 

	
\def \baselinestretch{.95}
\begin{abstract}
 The proliferation of connected devices through Internet connectivity presents both opportunities for smart applications and risks to security and privacy. It is vital to proactively address these concerns to fully leverage the potential of the Internet of Things. IoT services where one data owner serves multiple clients, like smart city transportation, smart building management and healthcare can offer benefits but also bring cybersecurity and data privacy risks. For example, in healthcare, a hospital may collect data from medical devices and make it available to multiple clients such as researchers and pharmaceutical companies. This data can be used to improve medical treatments and research but if not protected, it can also put patients' personal information at risk. To ensure the benefits of these services, it is important to implement proper security and privacy measures.
 In this paper, we propose a symmetric searchable encryption scheme with dynamic updates on a database that has a single owner and multiple clients for IoT environments. Our proposed scheme supports both forward and backward privacy. Additionally, our scheme supports a decentralized storage environment in which data owners can outsource data across multiple servers or even across multiple service providers to improve security and privacy. Further, it takes a minimum amount of effort and costs to revoke a client's access to our system at any time. The performance and formal security analyses of the proposed scheme show that our scheme provides better functionality, and security and is more efficient in terms of computation and storage than the closely related works\let\thefootnote\relax\footnotetext{*This version was submitted to ESORICS 2023.}.
  \end{abstract}





\section{Introduction}
\label{introduction}
With the rapid growth of the Internet of Things (IoT) and big data, massive amounts of data with different sources and environments are being produced worldwide per day ranging from health monitoring to smart buildings \cite{Abdulmalik2021}. The generated data are of great value and sensitive in nature, and it is essential to process, store, and manage these data securely and efficiently. Due to the advancements in Cloud Computing technologies, the generated huge volume of data can be stored and managed in remotely located storage servers, while getting easy accessibility, better availability, low initial investment costs, etc. \cite{Hayes2008}. As the data are now being outsourced to external storage servers, the data owners (who own the data) lose control over them and can no longer protect them like their own local machines. This gives full dependency on the third party, also commonly known as \emph{service provider} for safe keeping of the outsourced data. But it brings other challenges like trust and security capability issues associated with the service provider. As the outsourced data contains sensitive information, like electronic health records, and personal information like driving licence, car number plate, home addresses, etc., the service provider might itself want to gain access to this sensitive information for various motivations\footnote{The service provider might want to sell acquired sensitive information of the data owners to other interested parties for monetary benefits.}. Further, the regular occurrence of data breach incidents raises questions about the capabilities of the service providers to maintain full-proof data security and privacy mechanisms. One solution to resolve this issue is to encrypt the data by the data owners before outsourcing it to the service providers. This way the data owners can share their sensitive data with authorized users by sharing secret keys while keeping the data safe from unauthorized entities including the service provider. 
\par 
However, this process brings another set of challenges, including difficulty in keyword search. In a keyword search, the data owner (or any clients authorized by the data owner) should be able to retrieve the desired data from the cloud storage servers without revealing any sensitive information about the searched keywords and the associated data/files to the servers. Searchable Encryption (SE) is a promising cryptographic technique that enables the data owners to outsource encrypted data in the cloud storage servers while allowing the authorized clients to delegate keyword search capabilities over the encrypted data to the servers without revealing any sensitive information of the searched keywords and the actual plaintext data \cite{Song2000}. The first practical SE scheme was proposed by Song \emph{et al.} in \cite{Song2000}. Afterwards, many schemes have been proposed to address various security issues and functionalities \cite{Bosch2014}- typically they are based on \emph{Symmetric Searchable Encryption} (SSE) and \emph{Asymmetric Searchable Encryption} (ASE). The fundamental difference between SSE and ASE is the use of symmetric-key and public-key cryptographic primitives, respectively. 
It has been observed that although ASE schemes can provide better flexibility and query expressiveness, ASE schemes are computationally expensive due to the use of expensive public-key cryptographic operations.
As such, ASE-based schemes are not suitable for IoT environments due to the limited resources (i.e., computational and storage resources) of the IoT devices. 
The SSE schemes are considered more efficient and practical for IoT environments due to the use of lightweight symmetric-key cryptographic operations \cite{Bosch2014}. However, most of the SSE schemes, 
only support single-owner and single-client scenarios, where the data owner can only perform keyword search operations over his/her outsourced encrypted data. However, this category of schemes is not suitable for an IoT environment with numerous users, where the data owner allows multiple users to perform keyword search queries for accessing the shared data in the cloud server \cite{Cui2021}. We observe that the single-owner and multi-client-based SSE scheme is more suitable for such IoT environments. 
Further, most of the SSE schemes, 
consider static databases which means they can't be updated easily or require re-encryption and re-uploading of the encrypted files after the initial setup. There are several advantages to having a dynamic database that supports the addition and deletion of encrypted files. It provides more flexibility and supports more real-world applications. The downside, however, is that it also introduces a new set of security risks, since more data is exposed. Bost et al. \cite{Bost2017} introduced the concept of forward and backward privacy in order to capture leakage in a dynamic setting. The dynamic SSE schemes \cite{zuo2019dynamic,zuo2020forward,kermanshahi2020geometric} that support both forward and backward security properties are not intended for multi-client settings, hence there is no access control mechanism in place. 
In this paper, we propose an efficient and secure single-owner and multi-client SSE scheme that supports a dynamic encrypted database with both forward and backward privacy for IoT environments. The major contributions of our proposed scheme are as follows:
\begin{itemize}
    \item Our scheme supports single-owner and multi-client settings. It enables the IoT data owner to delegate keyword-level search authorization to more than one client efficiently. It employs lightweight cryptographic operations that make it ideal for IoT environments.
    \item Our scheme also supports a dynamic encrypted database, which enables the IoT data owner to add and delete files at any time. The database updating operation also preserves both forward and backward privacy (our scheme has minimal leakages).
    \item Our scheme supports a decentralized storage environment, where the IoT data owner can outsource data in multiple servers or even in multiple service providers for achieving higher-level of security and privacy\footnote{In general, IoT data originates from various sources. It is recommended that the generated data should be stored in a decentralized platform because of regulations and privacy concerns \cite{Stolpe2016}.}.
    \item User revocation is supported in our design. Most importantly, the revocation operation does not require any computationally expensive operations, including re-encryption of the database. 
    \item Our performance and formal security analyses show that our scheme is more efficient in terms of computation and communication overhead and provides better security and functionality than closely related schemes.
\end{itemize}
\section{Related Work}\label{related-work}
It was Song et al. in \cite{Song2000} who first created a practical privacy-preserving keyword search scheme.
Curmola et al. \cite{curtmola2011searchable} introduced the symmetric searchable encryption scheme based on inverted indexing in 2011. The P3 scheme of Shen et al. \cite{shen2018secure} provides intelligent encrypted data processing in IoT cloud systems. To manage the location relationship of multiple queried keywords over encrypted data, homomorphic encryption was used along with bilinear maps. 
Using IoT devices, Guo et al. \cite{guo2018secure} constructed secure searchable encryption for range search. To encrypt their data, they use homomorphic and order-preserving encryption (OPE) along with a secure index built from the k-dimensional tree. 
Dynamic searchable encryption was proposed in order to better match real-world scenarios. 
The work published by Lipsdonk et al. \cite{van2010computationally} proposes a computationally efficient scheme for searchable symmetric encryption with dynamic updates to which the number of updates is finite. As a worst-case scenario, the search time is linearly dependent on the database size.
The inverted index approach from \cite{curtmola2011searchable} was used by Kasra-Kamaranshahi et al. \cite{kamara2012dynamic} to achieve a sub-linear search time by creating an SSE scheme that supported dynamic updates. 
This is followed by the use of Red-Black trees to construct the secure index in \cite{kamara2013parallel}, which is capable of simultaneously allowing for keyword searching and data updating. 
Based on blind storage, Naveed et al. \cite{naveed2014dynamic} proposed a new dynamic system with less information leakage, and the cloud server cannot track how many files are stored. 

Because of dynamic data updates using the leakage profile of a scheme, file injection attacks \cite{zhang2016all} may occur. Thus, there is an even greater need for forward security for data clients. In order to ensure forward security, the newly injected files cannot match previous trapdoors when dynamic data is being updated. Consequently, forward security has become a necessary component of searchable encryption schemes. 
For the first time, Stefanov et al. \cite{stefanov2013practical} developed a dynamic SE scheme that achieves forward security. During the update of the search trapdoor, Bost \cite{bost2016ovarphiovarsigma} introduces a novel scheme called $\Sigma o\varphi$o$\varsigma$ that uses only one-way permutations. 
In spite of the efficiency of $\Sigma o\varphi$o$\varsigma$, the use of the public key primitive makes it more computationally demanding. 
Besides forward security, searchable encryption schemes require backward security as well. Backward security ensures the security of the database and its updates during search queries. 
As a general rule, search queries should not expose corresponding deleted documents.
By using constrained pseudo-random functions and puncturable encryption primitives, Bost et al. \cite{bost2017forward} presented several forward and backward secure SE schemes. 
A series of dynamic SSE schemes provided by Zuo et al. \cite{zuo2019dynamic,zuo2020forward} possesses both forward and backward security properties by combining different cryptographic primitives. 
In a dynamic SSE scheme, Kasra Kermanshahi et al. \cite{kermanshahi2020geometric} provided forward and backward security and geometric range search. 


\begin{figure}[t]
		\centering
		\fbox{\scalebox{3}{\includegraphics[width=1.5cm, height=.8cm]{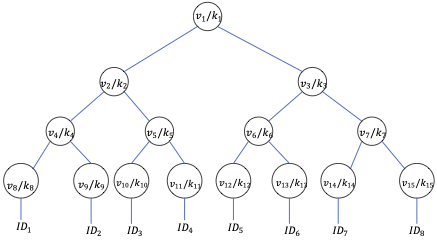}}} 
		\caption{Tree Based Group Key Distribution}
		\label{tree}
	\end{figure}
\section{Preliminaries}
\label{preli}
This section presents some of the utilized concepts in our scheme. The frequently used notations are shown in Table \ref{notation}. 
\begin{table}[h]
	\tabcolsep 1.0pt
	\centering
	\caption{NOTATIONS}
 \algsize{
	\begin{tabular}{p{1.8cm}p{10cm}}
		\hline
		Notation  & Description
		\\[0.5ex]    \hline
		$\mathbb{U}$, $\mathbb{A}$ & a set of clients and attribute universe respectively  \\\hline 
		$\mathcal{U}_i$ & set of clients authorized to access keyword $w_i$\\\hline
		$\mathbb{W}, \Gamma_{w}$ & set of keywords and having common documents with $w$ respectively\\\hline
		$\mathcal{H}, H_1$& chameleon hash function and Hash function $H_1: \{0, 1\}^*\rightarrow \mathbb{Z}_q^*$, respectively\\\hline
		$st^{dht}_{att_i}$& state of the distributed hash table for the $i^{th}$ attribute and keyword $w$, respectively\\\hline
		$st^j_w$ & $j^{th}$ state of the keyword $w$ in the database\\\hline
		$\mathtt{k}_{id_i}$, $\mathtt{k}_{att_i}$& $i^{th}$ client key and attribute key, respectively\\\hline
		$\mathtt{K}^{1w}_j, \mathtt{K}^{2w}_j, \mathtt{K}^{3w}_j$& $j^{th}$ version of the keyword key for $w$; $j$ starts from $0$\\\hline
	\end{tabular}}
	\label{notation}
\end{table}
\subsection{Group Key Distribution Method}
\label{key-distribution}
Our scheme uses a tree-based group key assignment scheme similar to the one proposed in \cite{Naor2001} for the clients. Figure \ref{tree} shows a binary tree for the universe of clients, $\mathbb{U}$. Each node in the tree represents a secret key, say $i^{th}$ node represents key $k_i$. We called it \emph{node key}. A set of node keys from the root to a leaf node is called \emph{path key}. Every client in $\mathbb{U}$ is assigned with a leaf node and associated \emph{path key} in the tree. For example, the client $\mathtt{ID}_1$ is assigned the path key, $\mathtt{pathkeys}_1= \{k_1, k_2, k_4, k_8\}$. Similarly, the client $\mathtt{ID}_6$ is assigned the path key, $\mathtt{pathkeys}_6= \{k_1, k_2, k_6, k_{13}\}$. Our key assignment method is secure, as keys are assigned randomly and independently from each other. More details can be found in \cite{Naor2001}.

\subsection{Symmetric Additive homomorphic encryption}
\label{Homomorphic-Encryption}
 Our scheme uses a slightly modified version of the additive homomorphic encryption scheme defined in \cite{Papadimitriou2016}. We replace the addition modulo $n$ with the exclusive-or, "XOR" ($\oplus$) operation for the addition modulo $2$ which can satisfy the requirements in our proposed scheme.
\par 
Let $\mathcal{F}: \{0, 1\}^\lambda\leftarrow K \times I$ be a pseudo-random function (PRF), where $\mathtt{K}$ be the set of the secret keys and  $I$ be the set of identifiers, and the plaintexts $m\in \mathbb{Z}_2: \{0, 1\}^{l}$. We define an additive homomorphic encryption scheme $\mathtt{E}:(\mathtt{Enc}, \mathtt{Dec})$ as follows:
\begin{align}
    C=& \mathtt{Enc_K}(m, i)= \left(\left(m\oplus \mathcal{F}(\mathtt{K}, i)\oplus \mathcal{F}(\mathtt{K}, i-1)\right), i\right)\\
    m=& \mathtt{Dec_K}(C, i)= C\oplus\mathcal{F}(\mathtt{K}, i)\oplus \mathcal{F}(\mathtt{K}, i- 1)
\end{align}
Let $(C_1, i-1)= \mathtt{Enc_K}(m_1, i-1)$ and $(C_2, i)= \mathtt{Enc_K}(m, i)$ are the two ciphertexts that we want to perform homomorphic addition operation, we have
\begin{align}
    (C_1, i- 1)\oplus(C_2, i)=& (C_1\oplus C_2, i):= \mathtt{Enc_K}\left((m_1\oplus m_2), i\right)
\end{align}

\subsection{Dynamic Searchable Symmetric Encryption (DSSE)}\label{sec:sse}
We follow the database model given in the paper \cite{bost2016ovarphiovarsigma}.
A database is a collection of  (index, keyword set) pairs 
denoted as $\mathbb{DB}=(ind_i,\mathbb{W}_i)_{i=1}^d$, where $ind_i\in\{0,1\}^{\ell}$ and $\mathbb{W}_i\subseteq\{0,1\}^*$. 
The set of all keywords of the database $\mathbb{DB}$ is $\mathbb{W}=\cup_{i=1}^d\mathbb{W}_i$, 
where $d$ is the number of documents in $\mathbb{DB}$. 
We identify $|\mathbb{W}|$ as the total number of keywords and 
$N=\Sigma_{i=1}^d|\mathbb{W}_i|$ as the number of document/keyword pairs. 
We denote $\mathbb{DB}(w)$ as the set of documents that contain a keyword $w$. 
To achieve a sublinear search time, 
we encrypt the file indices of $\mathbb{DB}(w)$ corresponding to the same keyword $w$ (a.k.a. inverted index).

A DSSE scheme $\Gamma$ consists of an algorithm $\mathrm{Setup}$ (($\mathbb{EDB}, \sigma) \leftarrow \mathrm{Setup}(\mathbb{DB}, 1^{\lambda}$)) and two protocols $\mathrm{Search}$ (($\mathcal{I}; \perp) \leftarrow \mathrm{Search}(q, \sigma; \mathbb{EDB})$) and $\mathrm{Update}$ ($(\sigma'; \mathbb{EDB}') \leftarrow \mathrm{Update}(\sigma, op, in; \mathbb{EDB})$). 

\subsubsection{DSSE Leakage Profile}

In this section, we define the general leakage functions, $\mathcal{L}$, associated with dynamic searchable symmetric encryption schemes \cite{bost2017forward}. 
\begin{itemize}
    \item $\mathrm{sp}(w)=\{ u:(u,w) \in Q \}$ is the search pattern which shows two search queries pertain to the same keyword, $w$. This leakage function records the list $Q$ of every search query, in the form $(u,w)$, where $u$ is the timestamp (increases with every query).
    \item $\mathrm{UpHist}(w)$ is a history which outputs the list of all updates on keyword $w$. Each element of this list is a tuple $(u, op, ind)$, where $u$ is the timestamp of the update, $op$ is the operation, and $ind$ is the updated index.
    \item $\mathrm{TimeDB}(w)$ is the list of all documents matching $w$, excluding the deleted ones, together with the timestamp of when they were inserted in the database.
    \item $\mathrm{Updates}(w)$ is the list of timestamps of updates on $w$.
\end{itemize}

\subsection{Chameleon-hash Functions}
\label{Chameleon hash}
Chameleon-hash functions \cite{krawczyk1998chameleon} also known as trapdoor-hash functions are the hash functions which have a trapdoor allowing one to find arbitrary collisions in the domain of the functions. However, as long as the trapdoor is not known, chameleon-hash functions are collision resistant. 
%
%
A chameleon-hash function $\mathrm{CH}$ consists of the following algorithms:

\begin{itemize}
\item $\mathrm{CH.Setup()} \xrightarrow{} (p, q, g, \mathrm{sk}, \mathrm{pk})$: This algorithm first chooses two large prime numbers $p$ and $q$ such that $p=kq+1$ for an integer $k$. Then, selects $g$ of order $q$ in $\mathbb{Z}^*_p$. Finally, it outputs $\xi \in \mathbb{Z}_q ^*$ as the private key $\mathrm{sk}$ and $y=g^{\xi} \mod p$ as the public key $\mathrm{pk}$.

 \item  $\mathrm{CH.Hash(x, r)} \xrightarrow{} g^xy^r \mod p$: On an input value $x$, a random value $r \in \mathbb{Z}_q ^*$ and outputs $H_{\mathrm{pk}}(x,r)=g^xy^r \mod p$.
 
 \item $\mathrm{CH.Forge(x, x', r)} \xrightarrow{} r'$ (Trapdoor collision): Given $x, x',r \in \mathbb{Z}^* _q$ as input, this algorithm outputs $r'$ such that $H_{\mathrm{pk}}(x,r)=H_{\mathrm{pk}}(x',r')$. This is done by solving for $r'$ in $x+\xi r=x'+\xi r' \mod q$.  
\end{itemize}



\begin{definition}[Indistinguishability]
For all pairs of message $x$ and $x'$, the probability distribution of the random value $H_{\mathrm{pk}}(x,r)$ and $H_{\mathrm{pk}}(x',r)$
are computationally indistinguishable.

\end{definition}

\begin{definition}[Collision-Resistance]
Without the knowledge of trapdoor key $sk$, there exists no efficient algorithm that, on input $x$, $x'$, and a random string $r$, outputs a string $r'$
that satisfy $H_{\mathrm{pk}}(x,r)=H_{\mathrm{pk}}(x',r')$, with non-negligible probability.

\end{definition}

\subsection{Bitmap Index}
\label{bit-string-representation}
Our scheme uses a notion called bitmap index, where a string is used to represent the presence of a keyword in a document. Bitmap index has been widely used in the database community as a special kind of data structure. 
A similar concept has also been used in \cite{Zuo2019}. In our scheme, each keyword, $w$ is associated with a bit string, $\mathtt{S}_{w_i}$ of length say $\gamma$, which is the maximum number of files that can be supported. Each bit in the bit string $\mathtt{S}_{w_i}$ represents a file in the database. If the $j^{th}$ bit of the bit string $\mathtt{S}_{w_i}$ of the keyword $w_i$ is $"1"$, it represents that the $j^{th}$ file in the database contains the keyword $w_i$. On the other hand, if the $j^{th}$ bit is $"0"$, it represents $j^{th}$ file in the database that does not contain the keyword $w_i$. To illustrate it further, let's assume $6$ files $f_0, f_1, \cdots, f_5$ in a database. Let's also assume that initially $w_1$ and $w_2$ keywords have the files $f_0, f_3, f_5$ and $f_0, f_2, f_4$ respectively. For example, the bit string representations ($100101$ for the keyword $w_1$ and $101010$ for the keywords $w_2$). Suppose, we want to delete the file $f_0$ from the keyword $w_1$. We can do it by flipping the $0^{th}$ bit in the $w_1$'s bit string $100101$ from "1" to "0". Similarly, if we want to add a file, say $f_3$ to the keyword $w_2$, we can flip the $4^{th}$ bit position of the $w_2$'s bit string $101010$ from $0$ to $1$. We can easily flip the bits in a bit string using the standard "Exclusive OR" (XOR) operations. We can generate an update bit string where the position of the bit to be changed is equal
to ``1"  and the rest are ``0”s. Finally, we can perform the  XOR operation between the original bit string and the updated bit string where $f_0$ is deleted using the update string $100000$ from the keyword $w_1$ and $f_3$ is added using the update string $000100$ to the keyword $w_2$. We can observe that we can easily encrypt the bit strings using the symmetric additive homomorphic encryption scheme as defined in Section \ref{Homomorphic-Encryption} and perform update operations to reflect the addition or deletion of files (due to the homomorphic property). 

\section{Our Proposed Scheme}
\label{proposed-scheme}
Our proposed scheme consists of four main phases, namely \emph{System Initialization}, \emph{Keyword Search}, \emph{Database Update}, and \emph{Client Revocation}. 
\begin{algorithm}[H]
\caption{System Initialization} 
        \label{Setup}
  \algsize  
\hspace*{\algorithmicindent} \textbf{Input} Security parameter $1^\lambda$, group key distribution tree $\mathcal{T}$, attribute universe $\mathbb{U}_{\mathbb{A}}$, keyword set $\mathbb{W}$, document identifier set $\mathbb{DB}(\mathbb{W})$\\
 \hspace*{\algorithmicindent} \textbf{Output} Public parameter $\mathtt{PP}$, master secret $\mathtt{MS}$, encrypted database $\mathbb{EDB}$
\begin{algorithmic}[1]
         \Algphase{Phase 1: Data Owner Setup \& Key Generation}
		\State Data owner chooses a random secret key $\mathtt{MK}\in \mathbb{Z}_q^*$ for PRF $\mathcal{F}$\;
		\For{each node $v_i$ in $\mathcal{T}$}
		    \State Data owner computes a node key, $k_i\leftarrow \mathcal{F}(\mathtt{MK}, v_i)$\;
		 \EndFor
		\For{each client $\mathtt{ID}_i$ in $\mathbb{U}$}
		\State Data owner computes client key $\mathtt{k}_{id_i}\leftarrow \mathcal{F}(\mathtt{MK}, id_i)$ \;
		\State Data owner selects a random public label $l_{id_i}\in \{0, 1\}^\lambda$\;
        \State Data owner sets a path key, $\mathtt{pathkey}_{id_i}= (k_1|| \dots|| k_j|| \dots|| k_{h})$\;
        \State Data owner computes a public path key token $\mathtt{PathKeyToken}_{id_i}= \mathtt{pathkey}_{id_i}\oplus \mathcal{F}(l_{id_i}, \mathtt{k}_{id_i})$\;
		\EndFor
        \Algphase{Phase 2: EDB Generation}
        \For{each $att_i$ in $\mathbb{A}$}
        \State Data owner computes $\mathtt{k}^{dht}_{att_i}\leftarrow \mathcal{F}(\mathtt{MK}, st^{dht}_{att_i} || att_i)$\;
        \State $S_{att_i}\leftarrow \mathcal{F}(\mathtt{k}^{dht}_{att_i}, att_i)$\;
        \State Append $S_{att_i}$ to DHT\;
        \EndFor
		\For{each $att_i\in \mathbb{A}$}
		    \State Data owner computes an attribute key $\mathtt{k}_{att_i}\leftarrow \mathcal{F}(\mathtt{MK}, att_i)$\;
		    \State Data owner initializes an empty map $\mathcal{D}_\mathbb{W}$\;
		        \For{each $w$ in $\mathbb{W}_{att_i}$}
                    \State Data owner selects a random public label $l_{w}\in \{0, 1\}^\lambda$\;
		            \State Data owner sets state of the keyword, $st^j_w= st^{j-1}_w+ 1$\;
		            \State Data owner computes keyword keys $\mathtt{K}^{1w}_j\leftarrow H_1(\mathcal{F}(\mathtt{k}_{att_i},w|| 0|| st^j_{w}))$, $\mathtt{K}^{2w}_j\leftarrow H_1(\mathcal{F}(\mathtt{k}_{att_i},w||1|| st^j_{w}))$, $\mathtt{K}^{3w}_j\leftarrow H_1(\mathcal{F}(\mathtt{k}_{att_i}, w||2|| st^j_{w}))$\;
		            \State Data owner chooses a random number $\mathbbm{r}_j^w\in \mathbb{Z}_q^*$\;
		            \State Data owner computes $add^j_w\leftarrow \mathcal{H}(\mathtt{K}^{1w}_j, \mathbbm{r}_j^w)$\;
		            \State Data owner compute $r^w_j \xleftarrow{}\mathrm{CH.Forge}(\mathtt{K}^{1w}_j, \mathtt{K}^{2w}_j, \mathbbm{r}_j^w)$;
		            \State Data owner generates a bit string $\mathtt{S}^{w}$ to represent all files associated with $w$ as defined in Section \ref{bit-string-representation}\;
		            \State Data owner encrypts $\mathtt{S}^w$ with $\mathtt{K}^{3w}_j$, i.e., $e^w\leftarrow \mathtt{Enc}_{\mathtt{K}^{3w}_j}(\mathtt{S}^{w}, st^j_w)$\;
		            \State Data owner appends $\{r^w_j, e^w\}$ to $\mathcal{D}_\mathbb{W}[add^j_w]$ \;
     	      \EndFor
     	\EndFor
     	   \State \Return{$\mathtt{PP}= \big<\text{DHT}, \mathcal{F}, \mathcal{H}, H_1, \{l_{id}, \mathtt{PathKeyToken}_{id}\}_{\forall \mathtt{id}\in \mathbb{U}}, \big>, \mathtt{MS}= \big<\mathtt{MK}, \{\mathbbm{r}_j^w\}_{\forall w\in \mathbb{W}}, \{l_w, w \in \mathbb{W}\}\big>$, $\mathbb{EDB}= \big<\mathcal{D}_{\mathbb{W}}\big>$}
 	\end{algorithmic}
 \end{algorithm}
\subsection{System Initialization}
\label{system-initialization}
Our proposed scheme starts with the System Initialization phase. The processing steps of this phase are shown in Algorithm \ref{Setup}. The main purpose of this phase is to generate public parameters $\mathtt{PP}$, master secrets $\mathtt{MS}$, client related keys, and encrypted database $\mathbb{EDB}$. The data owner initiates this phase after taking a security parameter $1^\lambda$, group key distribution tree $\mathcal{T}$, attribute universe $\mathbb{A}$, keyword set $\mathbb{W}$, and all document identifier set $\mathbb{D}$. It has two main sub-phases, which are presented next. 
\par 
\subsubsection{Data Owner Setup \& Key Generation}
In this sub-phase, the data owner generates its master secret key $\mathtt{MK}$, group key distribution tree $\mathcal{T}$, and assigns secret keys to the clients. First, the data owner computes a key, termed as node key, $k_i$ by taking the master key $\mathtt{MK}$ and node's identity $v_i$ as input to the PRF, $\mathcal{F}$ for each node in the tree $\mathcal{T}$. Secondly, the data owner generates a secret key, $\mathtt{k}_{id_i}$ for each registered client $\mathtt{ID}_i$ and sends it securely to the corresponding client. Thirdly, the data owner assigns a leaf node in the tree $\mathcal{T}$ with each registered client. The data owner then assigns the path-key $\mathtt{pathkey}_{id_i}$ (please refer to Section \ref{key-distribution}) to the client $\mathtt{ID}_i$. The data owner shares the path-keys of the clients by publishing public path-key tokens, $\mathtt{PathKeyToken}_{id_i}$. Note that, the client $\mathtt{ID}_i$ can recover the path key $\mathtt{pathkey}_{id_i}$ using his/her secret key $\mathtt{k}_{id_i}$ and public label $l_{id_i}$ (detail are given in the rest of this section). 

\subsubsection{EDB Generation}
In this sub-phase, the data owner mainly generates the encrypted database, $\mathbb{EDB}$ for the keyword set $\mathbb{W}$, and Distributed Hash Table (DHT) for enabling distributed search to the clients. Our scheme uses map data structure, $\mathcal{D}_\mathbb{W}$ to locate the encrypted files associated with the searched keywords in a database. The data owner first divides the data into different attribute groups, so that a similar set of data can be stored in the same servers for making search simpler. Note that the keyword search queries or any further processing like analytic over the IoT data are commonly conducted over certain or correlated attributes \cite{Yuan2019}. With this observation, our scheme uses a partitioning algorithm that divides the database based on attributes \cite{Yuan2017}. This eventually enables to store the encrypted data with the same attributes at the same server(s)/CSPs. The data owner uses DHT for enabling authorized clients to locate the appropriate server(s) of their desired data by storing the addresses of the server(s)/CSPs. To achieve it, the data owner generates the DHT which contains addresses, $S_{att_i}$ of the servers/CSPs that contain attribute category $att_i\in \mathbb{A}$. Note that, the addresses $S_{att_i}$ are computed using $\mathtt{k}^{dht}_{att_i}$ associated with the attribute category $att_i$ (please refer to Lines $12, 13$ in Algorithm \ref{Setup}), and $\mathtt{k}^{dht}_{att_i}$ will be shared with the authorized clients for enabling them to recompute the addresses. Whenever a client wants to search data of an attribute category, the client first gets the address(es) of the server(s)/CSPs from the DHT. This eventually provides the data owner with a distributed storage platform having the ability to search. 
\par 
Once the DHT is generated, the data owner encrypts the keywords and the associated file identifiers. The data owner first generates a unique random secret key, $\mathtt{k}_{att_i}$ for each attribute category, $att_i$. Let's assume that the data owner wants to store keyword $w$ of attribute category $att_i$ and associated file identifiers in the encrypted database $\mathbb{EDB}$. The data owner generates three keyword keys $(\mathtt{K}^{1w}_{j}, \mathtt{K}^{2w}_{j}, \mathtt{K}^{3w}_{j})$ for encryption purposes of the keyword $w$ using the attribute key, $\mathtt{k}_{att_i}$ (please refer to Line $22$ in Algorithm \ref{Setup}). Now the data owner first computes the location $add^j_w$ in the map $\mathcal{D}_{\mathbb{W}}$ to store encrypted keyword information. The data owner uses a random number $\mathbbm{r^w_j}$, keyword-key $\mathtt{K}^{1w}_{j}$ as input to the chameleon hash function $\mathcal{H}$ for computing the address $add^j_w$ in the map $\mathcal{D}_{\mathbb{W}}$ (please refer to Line $24$ in Algorithm \ref{Setup}). The data owner also computes $r^w_j$ using the properties of chameleon hash function by using $(\mathtt{K}^{1w}_{j}, \mathbbm{r^w_j})$ and $\mathtt{K}^{2w}_{j}$ as input (please refer to Line $25$ in Algorithm \ref{Setup}). The main reason for this step is to enable the server to check the revocation status of a client. If a client is able to produce the same output as $\mathcal{H}(\mathtt{K}^{1w}_{j}, \mathbbm{r^w_j})$, then the client is considered as authorized; otherwise, the client is revoked or unauthorized. We shall explain it in more detail in Section \ref{keyword-search1}. Afterward, the data owner generates the bit string $\mathtt{S}_w$, as described in Section \ref{bit-string-representation}, associated with the keyword and encrypts it using the homomorphic encryption algorithm (described in Section \ref{Homomorphic-Encryption}) with the keyword key $\mathtt{K}^{3w}_{j}$. Finally, the data owner appends the ciphertext $e^w$ and $r^w_j$ into the map in the location $\mathcal{D}_{\mathbb{W}}[add^j_w]$. 
 \begin{algorithm}[H]
\caption{Keyword Search}
        \label{keyword-search-algo}
  \algsize  
  \hspace*{\algorithmicindent} \textbf{Input} Master secret $\mathtt{MS}$, public parameter $\mathtt{PP}$, group key distribution tree $\mathcal{T}$, keyword $w_1$, encrypted database $\mathbb{EDB}$\\
 \hspace*{\algorithmicindent} \textbf{Output} Document identifier set $\mathbb{DB}$
\begin{algorithmic}[1] 
    \Algphase{Phase 1: Client Authorization}
    \For{each node $d$ in $\mathtt{RootsSubTrees}(\mathcal{U}_i)$}
        \State Data owner computes $\mathtt{PubToken}^{w_i}_d= (\mathtt{k}^{dht}_{att_i}|| \mathtt{K}^{1w_i}_j|| \mathtt{K}^{2w_1}_j|| \mathtt{K}^{3w_i}_j|| \mathbbm{r}_j^{w_i})\oplus \mathcal{F}(k_d, l_{w_i})$
        \State Data owner stores $\mathtt{PubToken}^{w_i}_d$ in its public bulletin board
    \EndFor
    \Algphase{Phase 2: Trapdoor Generation} 
       \State The client $\mathtt{ID}$ gets $\mathtt{PubToken}^{w_i}_{d}$ from the public bulletin board of the data owner
        \If {$\mathtt{ID}\in \mathcal{U}_i$} 
            \State The client $\mathtt{ID}$ can recover the keyword keys $(\mathtt{k}^{dht}_{att_i}|| \mathtt{K}^{1w_i}_j|| \mathtt{K}^{2w_i}_j|| \mathtt{K}^{3w_i}_j|| \mathbbm{r}_j^{w_i})\leftarrow \mathtt{PubToken}^{w_i}_{d}\oplus \mathcal{F}(k_d, l_{w_i})$ using his/her path key component $k_d$ 
            \State The client computes $S_{att_i}\leftarrow \mathcal{F}_1(\mathtt{k}^{dht}_{att_i}, att_i)$
            \State The client computes $\mathtt{trap1}= \mathcal{H}(\mathtt{K}_j^{1w_i}, \mathbbm{r}_j^{w_i}); \mathtt{trap2}= \mathtt{K}_j^{2w_i}$ 
        \EndIf
     \State Finally, the client sends the trapdoors $\mathtt{Trap}= \big<\mathtt{trap1}, \mathtt{trap2}\big>$ to the server $S_{att_i}$
    \Algphase{Phase 3: Search}  
    \State The server gets $\{r_j^{w_i}, e^{w_i}\}\leftarrow \mathcal{D}_{\mathbb{W}}[\mathtt{trap_1}]$
    \If{$\mathtt{trap_1}\neq \mathcal{H}(\mathtt{trap2}, r^{w_i}_j)$}
        \State Aborts
    \EndIf
     \If{$\mathtt{trap_1}== \mathcal{H}(\mathtt{trap2}, r^{w_i}_j)$}
     \State Server sends $e^{w_i}$ to the client
    \EndIf
        \State The client gets $S^{w_i}\leftarrow e^{w_i}$ using $\mathtt{K}^{3w_i}_j$
        \State The client sets a file identifier set $\mathcal{DID}^{w_i}$
        \State The client appends all the file identifiers into the set $\mathcal{DID}^{w_i}\leftarrow S^{w_i}$ 
    \State Finally, the client sends $\mathcal{DID}^{w_i}$ to the server to fetch the actual encrypted files
    \State \Return{$\mathbb{DB}(w_i)$}
 	\end{algorithmic}
 \end{algorithm}

\subsection{Keyword Search}
\label{keyword-search1}
The main goal of this phase is to enable the authorized clients to generate the trapdoors for their desired keywords and also enable the server to retrieve the requested file identifiers using the trapdoors on behalf of the clients. Algorithm \ref{keyword-search-algo} shows the processing steps of this phase. This phase has three sub-phases, \emph{Client Authorization}, \emph{Trapdoor Generation} and \emph{Search}, which are explained next.
\subsubsection{Client Authorization} 
\label{client-authorization}
The first challenge for the data owner is to provide proper security credentials to the authorized clients for performing the keyword search. In this sub-phase, the data owner shares secret keys with the authorized clients, so that they can compute proper trapdoors for the keyword search of their desired keywords. The data owner also shares the $\mathtt{k}^{dht}_{att_i}$ to enable the authorized client for locating the server(s) that stores the requested documents. Suppose, $\mathcal{U}_i\subseteq \mathbb{U}$ be the set of clients, that appears in the group key assignment tree $\mathcal{T}$, are authorized to access keyword $w_i$. The data owner selects root nodes, $\mathtt{RootsSubTrees}(\mathcal{U}_i)$ of the minimum cover sets in the group key assignment tree $\mathcal{T}$ that can cover all of the leaf nodes associated with the clients in $\mathcal{U}_i$. The data owner then computes a public-token $\mathtt{PubToken}^{w_i}_d$ for each of the nodes $d$ in the minimum cover set $\mathtt{RootsSubTrees}(\mathcal{U}_i)$ to share $\mathtt{k}^{dht}_{att_i}$, keyword keys $(\mathtt{K}^{1w_i}_j, \mathtt{K}^{2w_1}_j, \mathtt{K}^{3w_i}_j)$, and the secret random number $\mathbbm{r}_j^{w_i}$ (please refer to Line 2 in Algorithm \ref{keyword-search-algo}). 
\subsubsection{Trapdoor Generation} In this sub-phase, the clients compute trapdoors of their desired keywords. Suppose, an authorized client $\mathtt{ID}_u$ wants to perform the keyword search for the keyword $w_1$. The client first needs to recover the keys (\emph{dht-key} $\mathtt{k}^{dht}_{att_i}$, keyword keys $(\mathtt{K}^{1w_i}_j, \mathtt{K}^{2w_i}_j, \mathtt{K}^{3w_i}_j)$, and secret random number $\mathbbm{r}_j^{w_i}$) to generate the trapdoors from the public-token of the keyword $w_i$, i.e., $\mathtt{PubToken}^{w_i}_d$. Note that, if the client is authorized, the client will have a common node-key(s) with $\mathtt{RootsSubTrees}(\mathcal{U}_i)$, say the common node key is $k_d$. The client can recover the keys from the public-token $\mathtt{PubToken}^{w_i}_d$ by performing XOR operation (i.e., $\oplus$) using $\mathcal{F}(k_d, l_{w_i})$ (please refer to Line 7 in Algorithm \ref{keyword-search-algo}). For example, if $\mathcal{U}_1= \{\mathtt{ID}_1, \mathtt{ID}_2, \mathtt{ID}_3, \mathtt{ID}_4, \mathtt{ID}_7, \mathtt{ID}_8\}$ in Figure \ref{tree}, $\mathtt{ID}_4$ can recover the keyword keys using the node-key $k_2\in \mathtt{pathkeys}_4$. This also implies that any client $\mathtt{ID}_j\notin \mathcal{U}_i$ cannot recover the keys. After getting the keys, the client gets the address of the server, $S_{att_i}$ using the \emph{dht-key}, and then computes the trapdoor $\mathtt{Trap}= \big<\mathtt{trap1, trap2}\big>$ (please refer to Lines $9, 10, 11$ in Algorithm \ref{keyword-search-algo}). Finally, the client sends the trapdoor $\mathtt{Trap}$ to the server(s) associated with the address(es) of $S_{att_i}$ in the DHT.
\subsubsection{Search}
In this sub-phase, the server(s) performs the actual keyword search operation over the encrypted database. The server first gets the location of the ciphertexts of the searched keyword using the trapdoor $\mathtt{trap1}$ (please refer to Line $12$ in Algorithm \ref{keyword-search-algo}). It then verifies if the trapdoor is revoked by comparing $\mathtt{trap1}$ with $\mathcal{H}(\mathtt{trap2}, \mathbbm{r}_j^{w_1})$ (please refer to Line $13$ in Algorithm \ref{keyword-search-algo}). If successful, it means that the client is authorized; otherwise, it aborts the connection. The server then gets the associated ciphertext $e^{w_i}$ with the location $\mathcal{D}_{\mathbb{W}}[\mathtt{trap1}]$ and sends it back to the client. The client then can get the plaintext bit index, $S^{w_i}$ and get the actual file identifiers associated with the keyword $w_i$ (please refer to Line $19$ in Algorithm \ref{keyword-search-algo}). Finally, the client can send the desired plaintext file identifiers to the server, and the server in return sends the actual encrypted files.  
\begin{algorithm}[H]
\caption{Database Update}
        \label{update}
  \algsize  
\hspace*{\algorithmicindent} \textbf{Input} Document identifier $f$, encrypted database $\mathbb{EDB}$, keyword $w_i$, state $st^j_{w_i}$, public parameter $\mathtt{PP}$, master secret $\mathtt{MS}$\\
 \hspace*{\algorithmicindent} \textbf{Output} Updated encrypted database $\mathbb{EDB}'$
\begin{algorithmic}[1]
    \State Data owner computes a bit string $\mathtt{S}^{w_i}_{up}$ to reflect the update (i.e., either addition or deletion of a document $f$)\;
    \State Data owner updates the state of the keyword $w_i$, $st^{j+1}_{w_i}= st^j_{w_i}+ 1$\;
    \State Data owner computes $add^j_{w_i}= \mathcal{H}(\mathtt{K}^{1w_i}_j, \mathbbm{r}_j^{w_i})$\;
    \State Data owner computes a ciphertext $e^{w_i}_{up}\leftarrow \mathtt{Enc}_{\mathtt{K}^{3w_i}_j}(\mathtt{S}^{w_i}_{up}, st^{j+1}_{w_i})$ using the additive symmetric homomorphic encryption method as described in Section \ref{Homomorphic-Encryption}\;%
     \State Data owner sends $add^j_{w_i}$, $e^{w_i}_{up}$ to the server\;
     \State The server gets the ciphertexts $e^{w_i}$ associated with $\mathbb{D}_{\mathbb{W}}[add^j_{w_i}]$\;
     \State The server performs a homomorphic addition operation between $e^{w_i}$ and $e^{w_i}_{up}$, i.e., $e^{w_i}\leftarrow\mathtt{Enc}_{\mathtt{K}^{3w}_j}(\mathtt{S}^{w}, st^j_{w_i})\oplus\mathtt{Enc}_{\mathtt{K}^{3w}_j}(\mathtt{S}^{w}_{up}, st^{j+1}_{w_i})$\;
     \State The server updates $\mathbb{EDB}$ by replacing the old $e^{w_i}$ with the final output of the homomorphic addition operation, i.e., 
     \State \Return{$\mathbb{EDB}'$}
 	\end{algorithmic}
 \end{algorithm}
 \subsection{Database Update}
\label{database-update}
In this phase, the data owner updates the database when one or more files are added or deleted from a keyword. The processing steps of this phase are presented in Algorithm \ref{update}. Suppose, the data owner wants to update the keyword $w_i$ to perform either addition or deletion of one or more files. The data owner first computes a bit string $\mathtt{S}^{w_i}_{up}$ to reflect the changes (please refer to Section \ref{bit-string-representation}) and encrypts it with the latest keyword key $\mathtt{K}^{3w_i}_j$ using the homomorphic encryption (please refer to Section \ref{Homomorphic-Encryption}). Finally, the data owner sends the ciphertext and the location of the keyword $add^j_{w_i}$ to the server. The server then performs homomorphic addition operation with the existing ciphertext $e^{w_i}$ and replaces it with the output. 
 \begin{algorithm}[H]
\caption{Client Revocation}
        \label{client-revocation}
  \algsize  
\hspace*{\algorithmicindent} \textbf{Input} Revoked keyword $w_i$, (non revoked) authorized client set $\mathcal{U}_i'$, Group key distribution tree $\mathcal{T}$, master secret $\mathtt{MS}$, public parameter $\mathtt{PP}$, encrypted database $\mathbb{EDB}$\\
 \hspace*{\algorithmicindent} \textbf{Output} Fresh keyword public-token $\{\mathtt{PubToken}^{w_i}_d\}_{d \in \mathtt{RootsSubTrees}(\mathcal{U}'_i)}$, updated $\mathbb{EDB}'$
\begin{algorithmic}[1]
    \State Data owner finds the attribute $att$ corresponding to the keyword $w_i$
    \State Data owner changes the state of the keyword $w_i$, $st^{j+1}_{w_i}= st^j_{w_i}+ 1$
    \State Data owner computes fresh keyword keys 
    $\mathtt{K}^{1w_i}_{j+1}\leftarrow H_1(\mathcal{F}(\mathtt{k}_{att}, w_i || 0 || st^{j+1}_{w_i}))$, $\mathtt{K}^{2w_i}_{j+1}\leftarrow H_1(\mathcal{F}(\mathtt{k}_{att}, w_i|| 1 || st^{j+1}_{w_i})$
    \State Data owner computes $\mathbbm{r}_{j+1}^{w_i} \xleftarrow{} \mathrm{CH.Forge(\mathtt{K}^{1w_i}_j, \mathtt{K}^{1w_i}_{j+1}, \mathbbm{r}_j^{w_i})}$ and $r_{j+1}^{w_i} \xleftarrow{} \mathrm{CH.Forge}(\mathtt{K}^{1w_i}_j, \mathtt{K}^{2w_i}_{j+1}, \mathbbm{r}_j^{w_i})$ 
    \State Data owner selects fresh root nodes, $\mathtt{RootsSubTrees}(\mathcal{U}'_i)$ of the minimum cover sets for the members in $\mathcal{U}'_i$. 
    \State Data owner computes fresh keyword public-tokens $\mathtt{PubToken}^{w_i}_d$, where $d \in \mathtt{RootsSubTrees}(\mathcal{U}'_i)$ using the same process defined in Section \ref{client-authorization} and deletes the old ones
    \State Finally data owner sends the fresh random number set $r^{w_i}_{j+1}$ along with $add^j_{w_i}= \mathcal{H}(\mathtt{K}^{1w_i}_j, \mathbbm{r}^{w_i}_j)$ to the server
    \State The server updates $\mathbb{EDB}$ to $\mathbb{EDB}'$ by replacing the previous random number $r^{w_i}_{j}$ with the fresh set $r^{w_i}_{j+1}$ associated with $\mathcal{D}_{\mathbb{W}}[add^j_{w_i}]$
     \State \Return{$\{\mathtt{PubToken}^{w_i}_d\}_{d \in \mathtt{RootsSubTrees}(\mathcal{U}'_i)}$, $\mathbb{EDB}'$}
 	\end{algorithmic}
 \end{algorithm}
  
 \subsection{Client Revocation}
 \label{revocation}
 In this phase, the data owner revokes one or more clients. Suppose, the data owner wants to revoke one or more authorized clients from accessing a keyword, say $w_i$. Suppose, $\mathcal{U}'_i$ is the set of non-revoked authorized clients. The data owner first changes the state, $st^{j}_{w_i}$ of the keyword $w_i$ and computes a fresh set of keyword-keys $(\mathtt{K}^{1w_i}_{j+1}, \mathtt{K}^{2w_i}_{j+1})$. The data owner also computes two random numbers $(\mathbbm{r}^{w_i}_{j+1}, r^i_{j+1})$ using the properties of the chameleon-hash function $\mathcal{H}$, which will be used to verify the revocation status of the clients as described in Section \ref{keyword-search1} by the server. Afterward, the data owner finds a fresh set of root nodes $\mathtt{RootsSubTrees}(\mathcal{U}'_i)$ for the non-revoked clients and computes fresh public-tokens to share the updated keyword keys $(\mathtt{K}^{1w_i}_{j+1}, \mathtt{K}^{2w_i}_{j+1})$ and the random number $\mathbbm{r}^{w_i}_{j+1}$ with the non-revoked authorized clients. The data owner sends $r^i_{j+1}$ to the server, which replaces the old random number $r^i_{j}$ associated with $\mathcal{D}_{\mathbb{W}}[add^j_{w_i}]$. Note that, the data owner does not change the keyword key $\mathtt{K}^{3w_i}_j$ in a revocation process, as we assume that the revoke clients have already accessed the encrypted documents. 
 \section{Discussion}
 \subsection{Security}
 The proposed scheme supports Forward and Backward privacy, the security proofs are given in Appendix. 
 \begin{table}[]
	\centering
	\caption{Security and Functionally Comparison}
 \algsize{
\begin{tabular}{|l|c|c|c|c|}
\hline
\multirow{2}{*}{Schemes} & \multicolumn{1}{l|}{\multirow{2}{*}{\begin{tabular}[c]{@{}l@{}}Forward \\ Privacy\end{tabular}}} & \multicolumn{1}{l|}{\multirow{2}{*}{\begin{tabular}[c]{@{}l@{}}Backward\\ Privacy\end{tabular}}} & \multicolumn{1}{l|}{\multirow{2}{*}{Multi Client}}& \multicolumn{1}{l|}{\multirow{2}{*}{Revocation}} \\
                         & \multicolumn{1}{l|}{}                                                                            & \multicolumn{1}{l|}{}                                                                            & \multicolumn{1}{l|}{}                             & \\ \hline
SMSE   \cite{wang2017towards}                  & -                                                                                                &   -                                                                                               & Y                                            &   -   \\ \hline
MCSSE     \cite{kermanshahi2019multi}   & -    &     -    & Y    & Y \\ \hline
NI-MCSSE    \cite{sun2020non} & -   &    -   & Y   & - \\ \hline
FP-MCSSE    \cite{bakas2019multi} & Y  &  -  & Y & -\\  \hline
MFS         \cite{wang2018multi}  & Y  & -   & Y    & Y \\ \hline
MCFPSSE     \cite{gan2021towards} & Y & -  & Y  & - \\ \hline
      Our scheme                           & Y                                                                                             &   Y                                                                                          & Y               & Y    \\ \hline

\end{tabular}}
\label{table-func}
\end{table}
\begin{table*}[!t]
	\centering
	\caption{Computation complexity comparison}
  \algsize{

\begin{tabular}{|l|l|l|l|}
\hline
Schemes        & Update &   Search & Revocation   \\ \hline
SMSE        \cite{wang2017towards}             & \multicolumn{1}{c|}{-}                                   & $O\left(\left|R_{w}\right|\right) \cdot\left(2 t_{B P}+t_{H}\right)$   & \multicolumn{1}{c|}{-} \\ \hline
MCSSE       \cite{kermanshahi2019multi}             & \multicolumn{1}{c|}{-}                                   & $O\left(\left|R_{w}\right|\right) \cdot\left(t_{H}\right)$              & \multicolumn{1}{c|}{$t_{ks}$} \\ \hline
NI-MCSSE      \cite{sun2020non}           & \multicolumn{1}{c|}{-}                                   & $O\left(\left|R_{w}\right|\right) \cdot\left(t_{H}\right)$         &  \multicolumn{1}{c|}{-} \\ \hline
FP-MCSSE     \cite{bakas2019multi}            & \multicolumn{1}{l|}{$t_{G}+t_{E n c}+t_{P E n c}+t_{H}$} & $O\left(\left|U_{w}\right|\right) \cdot\left(t_{H}\right)$       &   \multicolumn{1}{c|}{-} \\ \hline
MFS          \cite{wang2018multi}            & \multicolumn{1}{l|}{$t_{B P}+t_{E}+3 t_{H}$}             & $O\left(\left|U_{w}\right|\right) \cdot\left(t_{B P}+t_{H}\right)$   & \multicolumn{1}{c|}{$**$} \\ \hline
MCFPSSE    \cite{gan2021towards}              & \multicolumn{1}{l|}{$2 t_{H}$}                           & $O\left(\left|U_{w}\right|\right) \cdot\left(t_{H}\right)$        &  \\ \hline
Our Scheme  & $2t_{H_1} + 2|\mathbb{U}| + t_{CH} $ &  $3t_{H_1} + 2t_{CH} $  &  $(4 + |\mathtt{RST}_{\mathcal{U}'_i}|)t_{H_1} $ \\ 
& $+ (|\mathbb{D}|/\lambda + 1) \cdot t_\mathcal{F}$ & $+ (|\mathbb{D}|/\lambda + 1) \cdot t_\mathcal{F}$ & $+ 4|\mathbb{U}| + t_{CH}+ 2t_{CHF}$ \\ \hline
\end{tabular}}
\label{tab cost-comp}
\begin{tablenotes}
    \item $**$: the revocation is done via a request to proxy/server to remove  the query key of the user, hence no computation is involved. $t_{ks}$: time for the generation of new key shares for the valid users. $t_{H}$, for standard hash function, $t_{G}$ for the invertible pseudorandom function, $t_{E n c}$ for the symmetric encryption, $t_{P E n c}$  for public-key encryption, $t_{B P}$ for the bilinear pairing, $t_{E}$ the exponentiation time, $t_{RST_i}$ for finding the $\mathtt{RootsSubTrees}(\mathcal{U}_i)$ which is $2|\mathbb{U}|$ in the worst case, $t_{KG}$ for key generation, $t_{H_1}$ for hash function $H_1$, $t_\mathcal{F}$ for PRF $\mathcal{F}$, $t_{CH}$ for computing Chameleon Hash, $t_{CHF}$ for finding a collision of Chameleon Hash CH.Forge() given the trapdoor, $t_{HEnc}$/$t_{HDec}$ for symmetric additive homomorphic encryption/decryption operation which is $(|\mathbb{D}|/\lambda + 1) \cdot t_\mathcal{F}$, $t_{AH}$ for symmetric additive homomorphic addition operation which is only XOR operation than can be treated as negligible. $\left|U_{w}\right|$ the number of update operations about keyword $w .|R|$ the size of the result. $|\mathbb{DB}|$ number of keyword-document pairs; $|\mathbb{W}|$  number of keywords in $\mathbb{DB}$; $|\mathbb{D}|$ denotes the number of documents, $|\mathbb{U}|$ number of users, |$\mathbb{A}$| number of attributes, $|\mathbb{S}|$ number of servers; $|n|$ and $|s|$ the length of No.Files $[w]$ and No.Search $[w]$ in FPMC-SSE \cite{bakas2019multi}, respectively. $N$ the number of registered clients. 
\end{tablenotes}
\end{table*}
 \begin{table*}[!t]
	\centering
	\caption{Storage and Communication overhead comparison}
 \algsize{
\begin{tabular}{|l|l|l|l|l|l|}
\hline
Schemes & Client & Server & Third Party  &  Comm (Update) & Comm (Search)        \\ \hline
SMSE        \cite{wang2017towards}             & \multicolumn{1}{l|}{$O(1)$}                     & \multicolumn{1}{l|}{$O(|\mathrm{DB}|)$} & \multicolumn{1}{c|}{-}  & \multicolumn{1}{c|}{-}      & $O(|R|)$\\ \hline
MCSSE       \cite{kermanshahi2019multi}              & $O(|R|)$                    & \multicolumn{1}{l|}{$O(1)$}                     & \multicolumn{1}{l|}{$O(|\mathrm{DB}|)$}  & \multicolumn{1}{c|}{-}      & $O(|R|)$  \\ \hline
NI-MCSSE      \cite{sun2020non}          & \multicolumn{1}{l|}{$O(1)$}                     & \multicolumn{1}{l|}{$O(|\mathrm{DB}|)$} & \multicolumn{1}{c|}{-}  & \multicolumn{1}{c|}{-}      & $O(|R|)$  \\ \hline
FP-MCSSE     \cite{bakas2019multi}            & \multicolumn{1}{l|}{$O(|W|) \cdot(|n|+|s|)$}    & \multicolumn{1}{l|}{$O(|\mathrm{DB}|)$} & $O(|W|)$                & \multicolumn{1}{l|}{$O(1)$} & $O(|R|)$ \\ \hline
MFS          \cite{wang2018multi}          & \multicolumn{1}{l|}{$O(1)$}                     & \multicolumn{1}{l|}{$O(|D B|+N)$}       & $O(|W|+N)$              & \multicolumn{1}{l|}{$O(1)$} & $O(|R|)$ \\ \hline
MCFPSSE    \cite{gan2021towards}              & \multicolumn{1}{l|}{$O(|W|) \cdot(2|\lambda|)$} & \multicolumn{1}{l|}{$O(|D B|+N)$}       & $O(1)$                & \multicolumn{1}{l|}{$O(1)$} & $O(|R|)$    \\ \hline
Our Scheme  &  $O(|\mathbb{S}| + |\mathbb{D}| \cdot |\mathbb{W}| + |\mathbb{U}| + |\mathbb{U}|\cdot \mathtt{log}|\mathbb{U}|)$  &  $O(|\mathbb{D}| \cdot |\mathbb{W}|)$&    \multicolumn{1}{c|}{-}     & $O(|\mathbb{D}|)$ & $O(|\mathbb{D}|)$    \\ \hline

\end{tabular}}
\label{tab cost-storage}

\end{table*}



 \subsection{Performance Analysis}
 \label{performance-analysis}
This section analyses the performance of our scheme. We start this section by providing a theoretical performance analysis in Section \ref{theory-performance} and then the implementation and experimental results in Section \ref{experiments}. 
\subsubsection{Theoretical Performance Analysis}
\label{theory-performance}
Table \ref{table-func}, Table \ref{tab cost-comp}, and Table \ref{tab cost-storage} illustrate the comparison among state-of-the-art works from a functionality, security, and performance perspective. A number of SSE schemes (e.g. \cite{liu2018multi,kermanshahi2019multi,sun2020non,wang2017towards}) have been proposed for multi-client use. An efficient SSE with conjunctive keyword search and fast decryption in multi-client settings was proposed by Wang et al. \cite{wang2017towards} thanks to the server-side match technique. Using the distributed key-homomorphic pseudorandom function (PRF), Kasra-Kermanshahi et al. \cite{kermanshahi2019multi} constructed a novel multi-client SSE. NIMC-SSE, designed by Sun et al. \cite{sun2020non}, does not require data owners to interact with clients in order to be efficient.  Their scheme leverages attribute-based encryption to control access to cloud data at a finer level. Unfortunately, these schemes do not address dynamic databases. As a result of integrating a semi-trusted proxy server, Wang et al. \cite{wang2018multi} have proposed a multi-client forward private SSE scheme with optimal search complexity. Nevertheless, this scheme requires bilinear pairings, resulting in heavy computation overhead.   
  A multi-client SSE scheme combining the use of an invertible PRF and trusted authority was proposed by Bakas and Michalas \cite{bakas2019multi}. However, there is no performance evaluation for a search and update protocol. 
\par
The following is the performance analysis of our approach:
\begin{itemize}
    \item Computation cost (Table \ref{tab cost-comp})
    \begin{itemize}
        \item For $\mathbb{EDB}$ generation, the computation cost at the data owner is $|\mathbb{W}| \cdot (3t_\mathcal{F} + 2t_{H_1} + t_{CH} + t_{CHF} + t_{HEnc}) + \sum_{i=1}^{|\mathbb{A}|}{t_{RST_i}} = |\mathbb{W}| \cdot ((4 + |\mathbb{D}|/\lambda)t_\mathcal{F} + 2t_{H_1} + t_{CH} + t_{CHF}) + 2|\mathbb{A}||\mathbb{U}|$
        \item For search operation, the computation complexity at a server is $t_{CH}$, at each user is $3t_{H_1} + t_{CH} + t_{HDec}$. Then, the total search computation complexity is $3t_{H_1} + 2t_{CH} + (|\mathbb{D}|/\lambda + 1) \cdot t_\mathcal{F}$
        \item For update operation, the complexity at the data owner is $2t_{H_1} + t_{RST_i} + t_{CH}+ t_{HEnc}$, at a server is $t_{AH}$. Thus, the total computation complexity is $2t_{H_1} + 2|\mathbb{U}| + t_{CH} + (|\mathbb{D}|/\lambda + 1) \cdot t_\mathcal{F}$
        \item For the user revocation operation, the computation complexity at the data owner is $(4 + |\mathtt{RootsSubTrees}_{\mathcal{U}'_i}|)t_{H_1} + 2t_{RST_i} + t_{CH}+ 2t_{CHF} = (4 + |\mathtt{RST}_{\mathcal{U}'_i}|)t_{H_1} + 4|\mathbb{U}| + t_{CH}+ 2t_{CHF}$; where $|\mathtt{RST}_{\mathcal{U}'_i}|$ is the number of nodes in $|\mathtt{RootsSubTrees}_{\mathcal{U}'_i}|$
    \end{itemize}
    \item Storage overhead (Table \ref{tab cost-storage})
    \begin{itemize}
        \item The storage overhead at the data owner is $O(|\mathbb{S}| + |\mathbb{D}| \cdot |\mathbb{W}| + |\mathbb{U}| + |\mathbb{U}|\cdot \mathtt{log}|\mathbb{U}|)$; including $\lambda \cdot |\mathbb{S}|$ bits for DHT, $\lambda \cdot |\mathbb{U}|$ bits for $\{l_{id}\}_{id \in \mathbb{U}}$, $\lambda \cdot |\mathbb{W}|$ bits for $\{l_{w}\}_{w \in \mathbb{W}}$, $\lambda \cdot |\mathtt{U}| \cdot (\ceil{\mathtt{log} |\mathtt{U}|} + 1)$ for $\mathtt{PathKeyToken}$, $\lambda_q \cdot (1 + |\mathbb{W}|)$ bits for ($MK, \{r^w\}_{w \in \mathbb{W}}$), $\lambda_p$ bits for $pk$, $\sum_{i=1}^{|\mathtt{W}|}{\lambda \cdot |\mathtt{RootsSubTrees}_{\mathcal{U}_i}|}$ for $\mathtt{PubToken}$, $(\lambda_q + |\mathbb{D}|) \cdot |\mathbb{W}|$ for $\mathbb{EDB}$, where $\lambda$ is the security bits of $\mathcal{F}$, $\lambda_q, \lambda_p$ are the security bits of large primes $p, q$ in Chameleon Hash
        \item The storage overhead at a server is $O(|\mathbb{D}| \cdot |\mathbb{W}|)$, which is $(\lambda_q + |\mathbb{D}|) \cdot |\mathbb{W}|$ bits of $\mathbb{EDB}$
        \item The storage overhead at each user is $O(1)$, which is $\lambda$ bits of $\mathtt{k}_{id}$
    \end{itemize}
    \item Communication cost (Table \ref{tab cost-storage})
    \begin{itemize}
        \item For $\mathbb{EDB}$ generation, the communication overhead is $O(|\mathbb{D}| \cdot |\mathbb{W}|)$, which is $(\lambda_q + |\mathbb{D}|) \cdot |\mathbb{W}|$ bits of $\mathbb{EDB}$ sent from the data owner to a server
        \item For search operation, communication overhead is $O(|\mathbb{D}|)$, which is ($\lambda_p + \lambda_q$) bits of $\mathtt{trapdoor}$ sent from a user to a server and $|\mathbb{D}|$ bits of the encryption search result sent back from the server to the client
        \item For update operation, communication complexity is $O(|\mathbb{D}|)$, which is $(\lambda_p + |\mathbb{D}|)$ bits of the address $add^j_w$ and the encryption $e^w_{up}$ of the updated inverted index sent from the data owner to a server
    \end{itemize}
\end{itemize}





\begin{figure}[t]
     \centering
     \subfloat[]{
         \includegraphics[width=2 in]{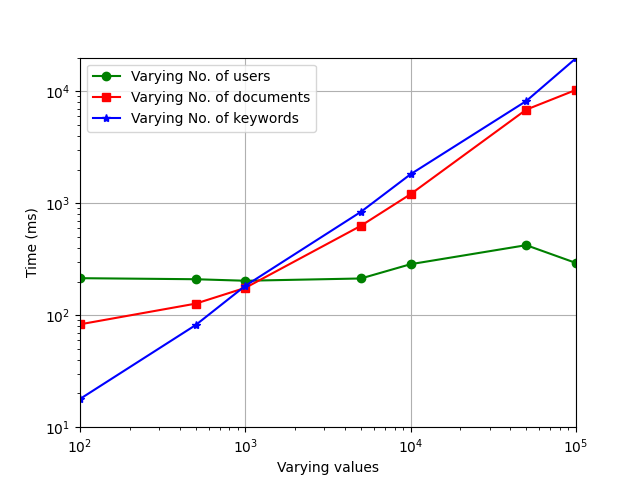}
          \label{time_dbgen}
     }
     \subfloat[]{
         \includegraphics[width=2 in]{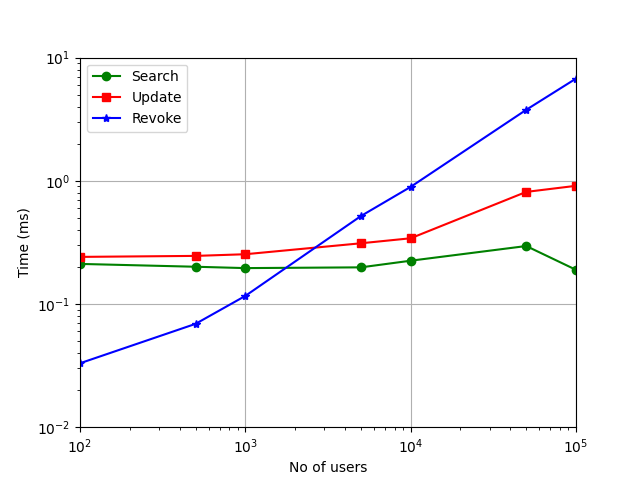}
          \label{varying_users}
     }
     \caption{Performance of our protocol: (a) Time of $\mathbb{EDB}$ generation w.r.t varying numbers of users/documents/keywords in milliseconds (b) Time of search, update, revoke w.r.t varying numbers of users in milliseconds at the setting $|\mathbb{W}| = 1000, |\mathbb{D}| = 1000$}
\end{figure}

\begin{figure}[t]
     \centering
   \subfloat[]{
         \includegraphics[width=2 in]{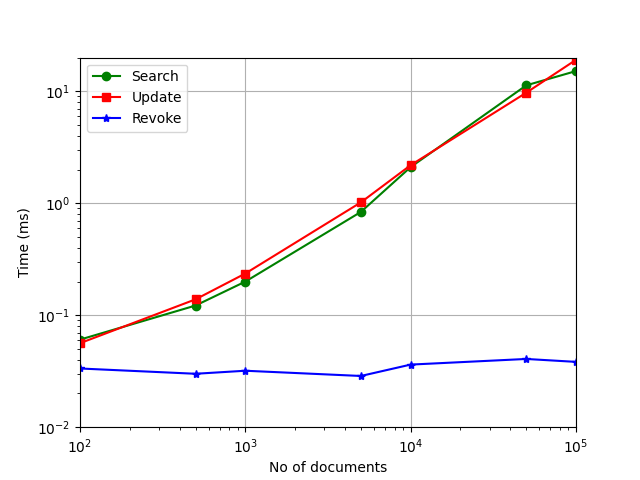}
          \label{varying_docs}
     }
    \subfloat[]{
     \includegraphics[width=2 in]{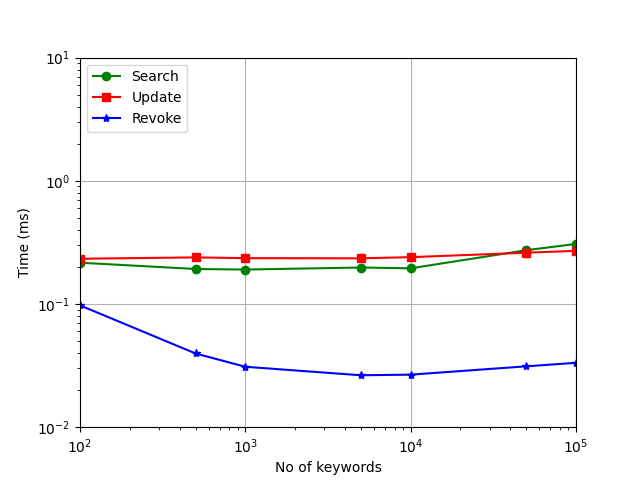}
      \label{varying_kws}
 }
     \caption{Performance of our protocol:  (a) Time of search, update, revoke w.r.t varying numbers of documents in milliseconds at the setting $|\mathbb{W}| = 1000, |\mathbb{U}| = 100$ (b) Time of search, update, revoke w.r.t varying numbers of keywords in milliseconds at the setting $|\mathbb{D}| = 1000, |\mathbb{U}| = 100$}
\end{figure}

\subsubsection{Implementation and Experimental Results}
\label{experiments}
This section presents the implementation of our scheme and results corresponding to different parameter settings.
\par 
The system is testing and deploying in Ubuntu Docker. The program is in C++ language, utilizing Pairing-Based Cryptography Library
(PBC) and openssl library. We test the performance of the system based on synthesis datasets of varying sizes, with different parameter settings of ($|\mathbb{D}|$, |$\mathbb{W}$|, |$\mathbb{U}$|). The size of a dataset, denoted as $|\mathbb{DB}|$ is the number of (keyword, id) pairs. For each set of parameters, we ran each operation (search, update, and revocation) 1000 times with random keywords and obtained the averaged performance in milliseconds. The setup deals with initial key generation and plain dataset generation. Keywords are classified in $\mathtt{log}|\mathbb{D}|$ types, in which type $i$ can access the set of files  $(0, 1, \cdots, |\mathbb{D}|/2^i)$. $\lambda = 128, \lambda_q = 160, \lambda_p = 1024$. The following is a more detailed interpretation of the implementation results.
\subsubsection{$\mathbb{EDB}$ generation}
As can be seen from Figure \ref{time_dbgen}, the time of $\mathbb{EDB}$ generation increases proportionally to the increase of the number of keywords or documents. Increasing the number of users does not change the time of encrypted dataset generation too much. For a medium dataset with $221 779$ {kewyord-document} pairs ($10^3$ keywords, $10^3$ documents), it just takes around 200 milliseconds to generate $|\mathbb{EDB}|$. For a bigger dataset, for example, with $22 199 779$ {keyword-document} pairs ($10^5$ keywords, $10^3$ documents), the time for generating the encrypted dataset is only 17672 milliseconds (<18 seconds). The implementation results demonstrate that the generation of $\mathbb{EDB}$ is efficient.
\subsubsection{Search}
Searching in the encrypted database of our approach is also efficient. Increasing the number of users or the number of keywords does not affect much on the performance of responding a random query. For example, the time answer remains around 0.2-0.3 milliseconds if increasing the number of users from $10^2$ to $10^5$, given $|\mathbb{W}| = 10^3, |\mathbb{D}| = 10^3$ (Figure \ref{varying_users}), or increasing the number of keywords from $10^2$ to $10^5$, given $|\mathbb{D}| = 10^3, |\mathbb{U}| = 10^2$ (Figure \ref{varying_kws}). However, when the number of documents rises, the time for search significantly increases (Figure \ref{varying_docs}) to around 15 milliseconds for the setting of $10^5$ documents, $10^3$ keywords, and 100 users.
\subsubsection{Update}
We examined the run-time performance of our proposed update strategy with our synthesis datasets and found that the update is also efficient too. Update time remains around 0.2-0.3 milliseconds when changing the number of keywords for the setting of 100 users, $10^3$ documents as can be seen from Figure \ref{varying_kws}. However, it is not the case when the number of users or documents is increased. For example, when the number of users increases to $10^5$, given $10^3$ keywords and $10^3$ documents, update time rises to nearly 1 millisecond (Figure \ref{varying_users}). Especially, the update time dramatically rises to around 18 milliseconds when increasing the number of documents to $10^5$, given $10^3$ keywords and 100 users. 
\subsubsection{Revocation}
Our approach offers an efficient solution for user revocation. There is no need to regenerate fresh user keys. The revocation time remains less than 0.1 milliseconds no matter how the number of words or documents changes, given the number of users is fixed at 100 users (Figure \ref{varying_kws}, Figure \ref{varying_docs}). When the number of users increases the revocation time also rises up due to the increasing time to find $\mathtt{RootsSubTrees}_{\mathcal{U}_i}$ and compute fresh public-tokens. For example, revocation time is around 7 milliseconds when having $10^5$ users.
\par
As can be seen from the implementation results, our systems can support large datasets, a big number of users, and still can provide efficient performance (in milliseconds) for search, update, and user revocation operations.

\section{Conclusion}
\label{conclusion}
We proposed a novel multi-client SSE scheme for the IoT environment that enables a data owner to delegate keyword search capabilities to multiple clients in an efficient and secure manner. Using DHT, we enable the data owner to store encrypted data efficiently across multiple servers or service providers. In addition, our scheme supports dynamic encrypted databases that help to add/delete files with minimal leakage ensuring both forward and backward privacy. Further, our scheme presented a client search privilege revocation method with minimal overhead in the system. We provided comprehensive security and performance analysis. The implementation results show that our scheme provides better functionality and outperforms the other notable works in terms of storage, computation, and communication overheads. 
 \bibliographystyle{unsrt}

 \appendix
 \section{Security Analysis}
 \label{security-analysis}
 \subsection{Security Definition}\label{subsec:security}
The security definition of DSSE is derived from two games: $\verb"DSSEREAL"_{\mathcal{A}}^{\Gamma}(1^{\lambda})$ and 
$\verb"DSSEIDEAL"_{\mathcal{A},\mathcal{S}}^{\Gamma}(1^{\lambda})$.
The game 
$\verb"DSSEREAL"_{\mathcal{A}}^{\Gamma}(1^{\lambda})$ 
is executed using DSSE.
The game  $\verb"DSSEIDEAL"_{\mathcal{A},\mathcal{S}}^{\Gamma}(1^{\lambda})$ 
is simulated using the leakage of DSSE. 
  The leakage can be described by a function $\mathcal{L}=(\mathcal{L}^{Stp}, \mathcal{L}^{Srch}, \mathcal{L}^{Updt})$, which describes what information is leaked to the adversary $\mathcal{A}$. Using the leakage function $\mathcal{L}$ as a measure of information leakage, if adversary $\mathcal{A}$ cannot discern these two games, then the information that can be inferred from it is the only information leaked. 
 In more formal terms,
\begin{itemize}
    \item $\verb"DSSEREAL"_{\mathcal{A}}^{\Gamma}(1^{\lambda})$: 
  	  On input a database $\mathbb{DB}$, 
 	   which is chosen by the adversary $\mathcal{A}$, 
  	  it outputs $\mathbb{EDB}$ by using $\mathrm{Setup}(1^{\lambda},\mathbb{DB})$ to the adversary $\mathcal{A}$. 
	$\mathcal{A}$ can repeatedly perform a search query $q$ (or an update query $(op,in$)).
	The game  outputs the results generated by running $\mathrm{Search}(q)$ 
	(or $\mathrm{Update}(op,in))$ to the adversary $\mathcal{A}$. 
	Eventually, $\mathcal{A}$ outputs a bit.
\item $\verb"DSSEIDEAL"_{\mathcal{A},\mathcal{S}}^{\Gamma}(1^{\lambda})$: 
On input a database which is chosen by the adversary $\mathcal{A}$, 
it outputs $\mathbb{EDB}$ to the adversary $\mathcal{A}$ by using 
a simulator $\mathcal{S}(\mathcal{L}^{Stp}(1^{\lambda}$, DB)). 
Then, it simulates the results for the search query $q$ by using the leakage function $\mathcal{S}(\mathcal{L}^{Srch}(q))$ 
and uses $\mathcal{S}(\mathcal{L}^{Updt}(op,in))$ to simulate the results for update query ($op,in$). Eventually, $\mathcal{A}$ outputs a bit.
\end{itemize}

\begin{definition}
A DSSE scheme $\Gamma$ is $\mathcal{L}$-adaptively-secure 
if for every PPT adversary $\mathcal{A}$, 
there exists an efficient simulator $\mathcal{S}$ such that $|Pr[\verb"DSSEREAL"_{\mathcal{A}}^{\Gamma}(1^{\lambda})=1]-Pr[\verb"DSSEIDEAL"_{\mathcal{A},\mathcal{S}}^{\Gamma}(1^{\lambda})=1]|\le negl(1^{\lambda})$.
\end{definition}
 \subsection{Forward Privacy}
The adaptive security of our construction relies on the semantic security of ASHE. All file indices are encrypted using ASHE. 
Without the secret key, the server cannot learn anything from the ciphertext. In our construction, for the update, we only leak the number of updates corresponding to the queried keywords \textbf{w}. Since all cryptographic operations are performed on the client side where no keys are revealed to the server, the server can learn nothing from the update, given that ASHE is IND-CPA secure. We can simulate the $\verb"DSSEREAL"$ as in Algorithm \ref{Sim2} and simulate the $\verb"DSSEIDEAL"$ by encrypting all 0's strings for $\mathbb{EDB}$. The adversary $\mathcal{A}$ can not distinguish the real ciphertext from the ciphertext of 0's. Then, $\mathcal{A}$ cannot distinguish $\verb"DSSEREAL"$ from $\verb"DSSEIDEAL"$. Hence, our Construction achieves forward security.
\begin{theorem}\label{th:fb}
(Adaptive forward privacy).  
Let $\mathcal{L}_{\Gamma}=(\mathcal{L}_{\Gamma}^{Srch}$, $\mathcal{L}_{\Gamma}^{Updt})$, 
where $\mathcal{L}_{\Gamma}^{Srch}(\textbf{w})=(sp(\textbf{w}))$, $\mathcal{L}_{\Gamma}^{Updt}(op,w,ind)=(Time(w))$,
 \textbf{w} is a set of queried keywords and $w\in\textbf{w}$, then our construction is $\mathcal{L}_{\Gamma}$-adaptively forward-private.
\end{theorem}




\begin{proof}

\textbf{Game} $G_{0}$: $G_{0}$ is exactly same as the real world game $\verb"DSSEREAL"_{\mathcal{A}}^{\Gamma}(1^{\lambda})$.$$Pr[\verb"DSSEREAL"_{\mathcal{A}}^{\Gamma}(1^{\lambda})=1]=Pr[G_{0}=1]$$

\textbf{Game} $G_{1}$: Instead of calling $F$ when generating $UT_w$, $G_{1}$ picks a new random key when it inputs a new keyword $w$, and stores it in a table $Key$ so it can be reused next time. If an adversary $\mathcal{A}$ is able to distinguish between $G_{0}$ and $G_{1}$, we can then build a reduction able to distinguish between $F$ and a truly random function. More formally, there exists an efficient adversary $\mathcal{B}_1$ such that $$Pr[G_{0}=1]-Pr[G_{1}=1]\le Adv_{F,\mathcal{B}_1}^{prf}(1^{\lambda}).$$

\begin{algorithm}[!htb]
    \caption{\textbf{Simulator} $\mathcal{S}_2$}\label{alg:bS}
    \label{Sim2}
     \algsize 
    \begin{multicols}{2}
    \underline{$\mathcal{S}.$\textbf{Setup}($1^{\lambda})$}
    \begin{algorithmic}[1]
        \State $K\leftarrow\{0,1\}^{\lambda}$
        \State $(SK)\leftarrow KeyGen(1^{\lambda})$
        \State \textbf{T} $\leftarrow$ empty map
        \State $m=0$
        \State \Return $(SK, K, \textbf{T}, m)$
    \end{algorithmic}

    \underline{$\mathcal{S}.$\textbf{Update}($Time(w)$)}\\
	\textit{Client:}
	\begin{algorithmic}[1]
            \State Parse $Time(w)$ as $(w,c)$
            \State $UT_w\leftarrow Key(w)$
            \For{$i=0$ to c}
            \State $e_{w,i}\leftarrow Enc(0's)$
            \State Send $(UT_w,e_{w,i})$ to the server.
            \EndFor	
\end{algorithmic}

    \underline{$\mathcal{S}.$\textbf{Search}($sp(w)$)}\\
    \textit{Client:}
    \begin{algorithmic}[1]
    \State $w\leftarrow sp(w)$
        	\State $UT_w\leftarrow Key(w)$
        	\State Send $UT_w$ to the server.
        \algstore{break}
    \end{algorithmic}
	
	\textit{Server:}
	\begin{algorithmic}[1]
		\algrestore{break}
		\State Upon receiving $UT_w$
		\State $e_w\leftarrow \textbf{T}[UT_w]$
		\State Send $e_w$ to the Client.
	\end{algorithmic}

	\end{multicols}
\end{algorithm}
\textbf{Simulator} We replace the bit string $bs$ with an all 0's string, we removed the useless part which will not influence the client's transcript. See Algorithm \ref{alg:bS} for more details. If an adversary $\mathcal{A}$ is able to distinguish between $G_{1}$ and $\verb"DSSEIDEAL"_{\mathcal{A},\mathcal{S}_2}^{\Gamma}(1^{\lambda})$, then we can build an adversary $\mathcal{B}_2$ to break the IND-CPA secure of ASHE. More formally, there exists an efficiently adversary $\mathcal{B}_2$ such that $$Pr[G_{1}=1]-Pr[\verb"DSSEIDEAL"_{\mathcal{A},\mathcal{S}_2}^{\Gamma}(1^{\lambda})=1]\le Adv_{\Sigma,\mathcal{B}_2}^{IND-CPA}(1^{\lambda}).$$

Finally, $$Pr[\verb"DSSEREAL"_{\mathcal{A}}^{\Gamma}(1^{\lambda})=1]-Pr[\verb"DSSEIDEAL"_{\mathcal{A},\mathcal{S}_2}^{\Gamma}(1^{\lambda})=1]$$$$\le Adv_{F,\mathcal{B}_1}^{prf}(1^{\lambda})+Adv_{\Sigma,\mathcal{B}_2}^{IND-CPA}(1^{\lambda})$$ which completes the proof.\hfill{$\Box$}

\end{proof}

\subsection{Backward privacy}
Every time we perform an update (all updates behave the same), the old ciphertext is replaced by a new one. Each ciphertext contains all the file identifiers as they were presented in a bit string. The result is that a)one cannot tell whether a component of the plaintext is updated, b) search queries do not reveal matching entries after they've been deleted.
Therefore, the proposed scheme meets the requirements of the "backward privacy" discussed below.

\begin{theorem}\label{th:bf}
 Let $F$ be a pseudo-random function and $Enc$ be a secure additive homomorphic symmetric encryption (ASHE), then our construction is $\mathcal { L }$-adaptively secure with the same leakage functions.
 



\end{theorem}

\begin{proof}
(Sketch) The construction B does not leak the type of update (either add or del) on encrypted file indices since it has been encrypted. 
Moreover, it does not leak file indices that have been previously added and/or deleted. 
The construction B is backward secure since the leakage is the same as in Theorem \ref{th:fb}.
The simulation follows the one from Theorem \ref{th:fb}.
\end{proof}
\end{document}